\documentclass{llncs}
\usepackage{amsmath}
\usepackage{amssymb}
\usepackage{graphicx}
\usepackage{stmaryrd}
\newcommand{\step}[1]{\stackrel{#1}{\longrightarrow}}
\newcommand{\primestep}[1]{\stackrel{#1}{\longrightarrow}\hspace{-3pt}'}
\newcommand{\bleft}{\texttt{[}}
\newcommand{\bright}{\texttt{]}}
\newcommand{\tless}{\texttt{<}}
\newcommand{\tmore}{\texttt{>}}
\newcommand{\inv}[1]{\square\,#1}
\newcommand{\diam}[1]{\Diamond\,#1}
\newcommand{\dlf}{\mathit{dlf}}
\newcommand{\pbis}{\ensuremath{\precapprox}}
\newcommand{\red}[1]{\stackrel{#1}{\rightarrowtriangle}}
\newcommand{\sub}[1]{\mathit{sub}\,(#1)}
\newcommand{\ispart}[2]{\mathit{part}\,(#1,#2)}
\title{Maximally Permissive Controlled System \\ Synthesis for Modal Logic
  \thanks{Supported by the EU FP7 Programme under grant agreement no. 295261 (MEALS).}}
\author{A.C. van Hulst \and M.A. Reniers \and W.J. Fokkink}
\institute{Eindhoven University of Technology, The Netherlands}
\begin{document}
\maketitle
\begin{abstract}
We propose a new method for controlled system synthesis on
non-deterministic automata, which includes the synthesis for
deadlock-freeness, as well as invariant and reachability
expressions. Our technique restricts the behavior of a 
Kripke-structure with labeled transitions, representing the
uncontrolled system, such that it adheres 
to a given requirement specification in an expressive modal logic. 
while all non-invalidating behavior is retained. This induces 
maximal permissiveness in the context of supervisory control. 
Research presented in this paper allows a system model to be 
constrained according to a broad set of liveness, safety and 
fairness specifications of desired behavior, and embraces
most concepts from Ramadge-Wonham supervisory control, including
controllability and marker-state reachability. Synthesis is defined
in this paper as a formal construction, which allowed a careful
validation of its correctness using the Coq proof assistant. 
\end{abstract}
\section{Introduction}
\label{sec:introduction}
This paper presents a new technique for controlled system synthesis
on non-deterministic automata for requirements in modal logic. The
controlled systems perspective treats the system under control --- the
\emph{plant} --- and a system component which restricts the plant
behavior --- the \emph{controller} --- as a single integrated entity.
This means that we take a model of all possible plant behavior, and
construct a new model which is constrained according to a logical
specification of desired behavior --- the \emph{requirements}. The
automated generation, or \emph{synthesis}, of such a restricted
behavioral model incorporates a number of concepts from supervisory
control theory \cite{cassandras}, which affirm the generated model as being a proper
controlled system, in relation to the original plant specification.
Events are strictly partitioned into being either controllable or
uncontrollable, such that synthesis only disallows events of the
first type. In addition, synthesis preserves all behavior which
does not invalidate the requirements, thereby inducing maximal
permissiveness \cite{cassandras} in the context of supervisory control. The requirement
specification formalism extends Hennessy-Milner Logic \cite{hml} with invariant,
reachability, and deadlock-freeness expressions, and is also able to express
the supervisory control concept of marker-state reachability \cite{ramadge}.

The intended contribution of this paper is two-fold. First, it presents
a technique for controlled system synthesis in a non-deterministic 
context. Second, it defines synthesis for a modal logic which is
able to capture a broad set of requirements.

Regarding the first contribution, it should be noted that supervisory
control synthesis is often approached using a deterministic model of
both plant and controller. Notably, the classic Ramadge-Wonham supervisory
control theory \cite{ramadge} is a well-researched example of this setup. 
The resulting controller restricts the behavior of the deterministic
plant model, thereby ensuring that it operates according to the
requirements via event-based synchronization. A controlled system
can not be constructed in this way for a non-deterministic model, as
illustrated by example in Fig. \ref{fig:printer}. Assume that
we wish to restrict all technically possible behavior of an indicator
light of a printer (Fig. \ref{fig:printer}a) such that after a single
$\mathit{refill}$ event, the indicator light turns $\mathit{green}$ 
immediately. In the solution shown in Fig. \ref{fig:printer}b, the
self-loop at the right-most state is disallowed, as indicated using
dashed lines, while all other behavior is preserved. Note that it is
not possible to construct this maximally-permissive solution using
event-based synchronization, as shown in \cite{paco}. However, an outcome as shown in Fig.
\ref{fig:printer}b can be obtained by applying synthesis for the 
property $\inv{\bleft\mathit{refill}\bright}\mathit{green}$, using
the method described in this paper. As this example clearly shows, the
strict separation between plant and controller is not possible for
non-deterministic models, and therefore we interpret the controlled
system as a singular entity.

\begin{figure}
\includegraphics[scale=.65]{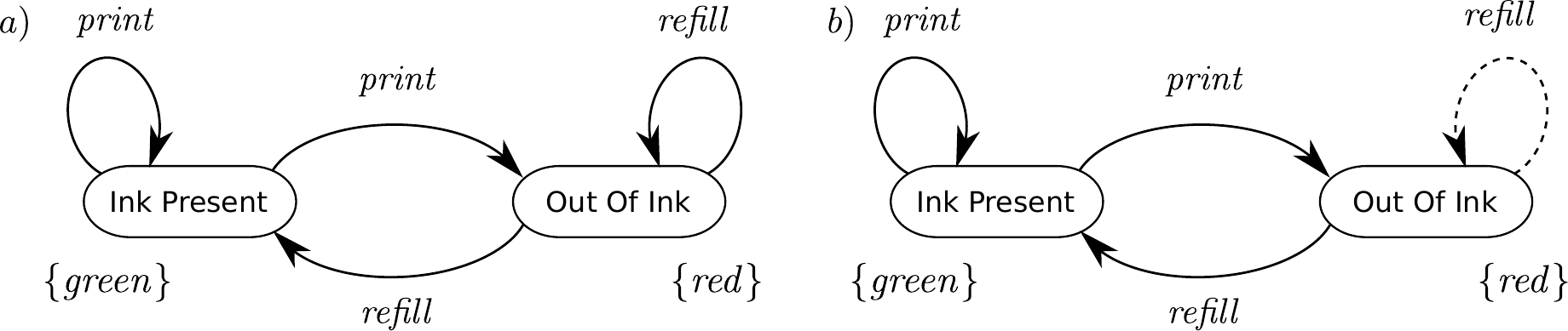}
\caption{Example of control synthesis in a non-deterministic context.
  A model for all possible behavior of an ink presence indicator light 
  of a printer is restricted in such a way that after every 
  $\mathit{refill}$, the state labeled with $\mathit{green}$ is reached
  directly. Synthesis, as defined in this paper, of the property 
  $\square\,\texttt{[}\mathit{refill}\texttt{]}\mathit{green}$ upon
  the model in Fig. \ref{fig:printer}a, results in a synthesis
  outcome as in Fig. \ref{fig:printer}b., where disallowed behavior
  is indicated using dashed lines.}
\label{fig:printer}
\end{figure}

The synthesized requirement in Fig. \ref{fig:printer}b represents a
typical example of a requirement in modal logic applied in this paper.
This requirement formalism, which extends Hennessy-Milner Logic with
invariant and reachability operators, and also includes a test for
deadlock-freeness, is able to express a broad set of liveness, safety,
and fairness properties. For instance, an important liveness concept
in supervisory control theory involves marker-state reachability,
which is informally expressed as the requirement that it is always
possible to reach a state which is said to be \emph{marked}. This
requirement is modeled as $\inv{\diam{\mathit{marked}}}$, using the
requirement specification logic, in conjunction with assigning $\mathit{marked}$ 
as a separate property to the designated states in the Kripke-model.

Safety-related requirements, which model the
absence of faulty behavior, include deadlock-avoidance,
expressed as $\inv{\dlf}$ (i.e., invariantly, deadlock-free) and
safety requirements of a more general nature. For instance, one
might require that some type of communicating system is always able
to perform a $\mathit{receive}$ step, directly after every 
$\mathit{send}$ step. Such a property is expressed as
$\inv{\bleft\mathit{send}\bright\tless\mathit{receive}\tmore
\mathit{true}}$, using the requirement specification logic applied in this paper.
In addition, we argue that this logic is able to model a limited class of
of fairness properties. One might require from a system which
uses a shared resource that in every state, the system has
access to the resource (the state has the $\mathit{access}$ property),
or it can do a $\mathit{lock}$ step to claim the resource, after which
access is achieved immediately. In order to constrain the behavior of the plant
specification such that it adheres to this requirement, we synthesize
the property $\inv{(\mathit{access}\vee\tless\mathit{lock}\tmore
\mathit{access})}$. 

The remainder of this paper is set up as follows. We consider a number
of related works on control synthesis in Section \ref{sec:related}. 
Preliminary definitions in Section \ref{sec:definitions} introduce
formal notions up to a formal statement of the synthesis problem.
Section \ref{sec:synthesis} concerns the formal definition of the
synthesis construction while Section \ref{sec:correctness} lists
a number of important theorems indicating correctness of the synthesis
approach, including detailed proofs, while these proofs are available
in computer-verified form as well \cite{sources}.

\section{Related Work}
\label{sec:related} 
Earlier work by the same authors concerning synthesis for modal logic
includes a recursive synthesis method for Hennessy-Milner Logic \cite{acsd},
and a synthesis method for a subset of the logic considered in this paper,
with additional restrictions on combinations of modal operators \cite{wodes}.

We analyze related work along three lines: 1) Allowance of non-determinism
in plant specifications, 2) Expressiveness of the requirement specification
formalism, and, 3) Adhering to some form of maximal permissiveness. Based
on this comparison, we analyze related work alongside the intended
improvements in this paper.

Ramadge-Wonham supervisory control \cite{ramadge} defines a 
broadly-embraced meth-odology for controller synthesis on 
deterministic plant models for requirements specified using
automata. It defines a number of key elements in the relationship
between plant and controlled system, such as controllability, 
marker-state reachability, deadlock-freeness and maximal permissiveness.
Despite the fact that a strictly separated controller
offers advantages from a developmental or implementational point
of view, we argue that increased abstraction and flexibility justifies
research into control synthesis for non-deterministic
models. In addition, we emphasize that the automata-based description
of desired behavior in the Ramadge-Wonham framework \cite{ramadge} 
does not allow the specification of requirements of existential nature.
For instance, in this framework it is not possible to specify that a step labeled with a particular
event \emph{must} exist, hence the choice of modal logic as our
requirement formalism.

Work by Pnueli and Rosner \cite{pnueli} concerns a treatment of synthesis
for reactive systems, based upon a finite transducer model of the plant,
and a temporal specification of desired behavior. This synthesis
construction is developed further for deterministic automata in \cite{pnueli},
but the treatment remains non-maximal. This research is extended 
in \cite{games}, which connects reactive synthesis to Ramadge-Wonham
supervisory control using a parity-game based approach. The methodology
described in \cite{games} transforms the synthesis control problem for
$\mu$-calculus formulas in such a way that the set of satisfying models of a
$\mu$-calculus formula coincides with the set of controllers which enforce
the controlled behavior. Although non-determinism is allowed in 
plant-specifications in \cite{games}, the treatment via loop-automata does 
not allow straightforward modeling of all (infinite) behaviors. Also,
maximal permissiveness is not specified as a criterion for control synthesis
in \cite{games}. Interesting follow-up research is found in \cite{nonnon},
for non-deterministic controllers over non-deterministic processes. However,
the specification of desired behavior is limited to
alternating automata \cite{nonnon}, which do not allow complete coverage of invariant
expressions over all modalities, or an equivalent thereof. Reactive synthesis
is further applied to hierarchical \cite{aminof} and recursive \cite{recursive}
component-based specifications. These works, which both are based upon a 
deterministic setting, provide a quite interesting setup from a developmental
perspective, due to their focus on the re-usability of components. 

Research in \cite{combining} relates Ramadge-Wonham supervisory
control to an equivalent model-checking problem, resulting in important
observations regarding the mutual exchangeability and complexity
analysis of both problems. Despite the fact that research in \cite{combining}
is limited to a deterministic setting, and synthesis results are not
guaranteed to be maximally permissive, it does incorporate a quite
expressible set of $\mu$-calculus requirements. Other research based upon
a dual approach between control synthesis and model checking studies the
incremental effects of transition removal upon the validity of $\mu$-calculus
formulas \cite{incremental}, based on \cite{cleaveland}. 

Research by D'Ippolito and others \cite{ippolito}, \cite{nonanomalous}
is based upon the framework of the world machine model for the synthesis of
liveness properties, stated in fluent temporal logic. A distinction is made
between controlled and monitored behavior, and between system goals and
environment assumptions \cite{ippolito}. A controller is then derived
from a winning strategy in a two-player game between original and required
behavior, as expressed in terms of the notion of generalized reactivity,
as introduced in \cite{ippolito}. Research in \cite{ippolito} also emphasizes the fact that
pruning-based synthesis is not adequate for control of non-deterministic
models, and it defines synthesis of liveness goals under a maximality
criterion, referred to as best-effort controller. However, this maximality
requirement is trace-based and is therefore not able to signify inclusion
of all possible infinite behaviors. In addition, some results in \cite{ippolito}
are based upon the assumption of a deterministic plant specification.

\section{Definitions}
\label{sec:definitions}
We assume a set $\mathcal{E}$ of events and a set 
$\mathcal{P}$ of state-based properties. In addition, we assume a strict
partition of $\mathcal{E}$ into controllable events $\mathcal{C}$ 
and uncontrollable events $\mathcal{U}$, such that $\mathcal{C}\cup
\mathcal{U}=\mathcal{E}$ and $\mathcal{C}\cap\mathcal{U}=\emptyset$.
State-based properties are used to capture state-based information, and
are assigned to states using a labeling function. Example properties
are shown in Fig. \ref{fig:printer}, as $\mathit{red}$ and 
$\mathit{green}$. Fig. \ref{fig:printer} also shows examples of
the events $\mathit{print}$ and $\mathit{refill}$, which are assumed
to be controllable in this example. Events are used to capture
system dynamics, and represent actions occurring when the system
transitions between states. Controllable events may be used to model
actuator actions in the plant, while an uncontrollable event may
represent, for instance, a sensor reading. Basic properties and
events are used to model plant behavior in the form of a 
Kripke-structure \cite{bull} with labeled transitions, to be 
abbreviated as Kripke-LTS, as formalized in Definition \ref{def:K}.
Note that we assume finiteness of the given transition relation. 

\begin{definition}
\label{def:K}
We define a Kripke-LTS as a four-tuple $(X,L,\longrightarrow,x)$ for
state-space $X$, labeling function $L:X\mapsto 2^\mathcal{P}$, finite transition
relation $\longrightarrow\,\subseteq X\times\mathcal{E}\times X$, and
initial state $x\in X$. The universe of all Kripke-LTSs is denoted by 
$\mathcal{K}$. 
\end{definition}

As usual, we will use the notation $x\step{e}x'$
to denote that $(x,e,x')\in\longrightarrow$. The reflexive-transitive closure
$\longrightarrow^*$ of a transition relation $\longrightarrow$ is
defined in the following way: For all $x\in X$ it holds that $(x,x)\in
\longrightarrow^*$ and if there exist $e\in\mathcal{E}$ and
$y,x'\in X$ such that $x\step{e}y$ and $y\longrightarrow^*x'$ then $(x,x')\in
\longrightarrow^*$. 

Two different behavioral preorders are applied in this paper. The first
is the simulation preorder, which is reiterated in Definition \ref{def:sim}.
Simulation is used to signify inclusion of behavior, while synthesis may
alter the transition structure due to, for instance, unfolding. Simulation
as applied in this paper is a straightforward adaptation of the definition
of simulation in \cite{glabbeek}. 

\begin{definition}
\label{def:sim}
For $k'=(X',L',\longrightarrow',x')$ and $k=(X,L,\longrightarrow,x)$ we say that $k'$
and $k$ are related via simulation (notation: $k'\preceq k$) if there exists a relation $R\subseteq
X'\times X$ such that $(x',x)\in R$ and for all $(y',y)\in R$ the following holds:
\begin{enumerate}
\item We have $L'(y')=L(y)$; and
\item If $y'\primestep{e}z'$ then there exists a step $y\step{e}z$ such that 
      $(z',z)\in R$. 
\end{enumerate}
\end{definition}

Partial bisimulation \cite{paco} is an extension of simulation such that
the subset of uncontrollable events is bisimulated. For plant specification
$k\in\mathcal{K}$ and synthesis result $s\in\mathcal{K}$ we require that
$s$ is related to $k$ via partial bisimulation. This signifies the fact that
synthesis did not disallow any uncontrollable event, which implies controllability
in the context of supervisory control. Research in \cite{paco} details the 
nature of this partial bisimulation preorder.

\begin{definition}
\label{def:pbis}
If $k'=(X',L',\longrightarrow',x')$ and $k=(X,L,\longrightarrow,x)$, 
then $k'$ and $k$ are related via partial bisimulation
(notation: $k'\pbis k$) if there exists a relation 
$R\subseteq X'\times X$ such that $(x',x)\in R$ and for all 
$(y',y)\in R$ the following holds:
\begin{enumerate}
\item We have $L'(y')=L(y)$;
\item If $y'\primestep{e}z'$ then there exists a step $y\step{e}z$ such that $(z',z)\in R$; and
\item If $y\step{e}z$ for $e\in\mathcal{U}$ then there exists a step $y'\primestep{e}z'$ such that $(z',z)\in R$.
\end{enumerate}
\end{definition}

Requirements are specified using a modal logic $\mathcal{F}$ given in Definition
\ref{def:F}, which is built upon the set of state-based formulas $\mathcal{B}$ in Definition
\ref{def:B}. 

\begin{definition}
\label{def:B}
The set of state-based formulas $\mathcal{B}$ is defined by the grammar:
\end{definition}
\begin{center}
$\mathcal{B}\mathbin{\texttt{::=}}\mathit{true}\mid\mathit{false}\mid\mathcal{P}\mid
  \neg\mathcal{B}\mid\mathcal{B}\wedge\mathcal{B}\mid\mathcal{B}\vee\mathcal{B}$
\end{center}

As indicated in Definition \ref{def:B}, state-based formulas are constructed
from a straightforward Boolean algebra which includes the basic expressions
$\mathit{true}$ and $\mathit{false}$, as well as a state-based property
test for $p\in\mathcal{P}$. Formulas in $\mathcal{B}$ are then combined using the 
standard Boolean operators $\neg$, $\wedge$ and $\vee$.

\begin{definition}
\label{def:F}
The requirement specification logic $\mathcal{F}$ is defined by the grammar:
\end{definition}
\begin{center}
$\mathcal{F}\mathbin{\texttt{::=}}\mathcal{B}\mid
  \mathcal{F}\wedge\mathcal{F}\mid\mathcal{B}\vee\mathcal{F}\mid
  \bleft\mathcal{E}\bright\mathcal{F}\mid\tless\mathcal{E}\tmore\mathcal{F}\mid
  \inv{\mathcal{F}}\mid\diam{\mathcal{B}}\mid\dlf$
\end{center}

We briefly consider the elements of the requirement logic $\mathcal{F}$. Basic
expressions in Definition $\ref{def:B}$ function as the basic building blocks
in the modal logic $\mathcal{F}$.  Conjunction is
included, having its usual semantics, while disjunctive formulas are restricted
to those having a state-based formula in the left-hand disjunct. This restriction
guarantees correct synthesis solutions, since it enables a local state-based
test for retaining the appropriate transitions.
The formula $\bleft e\bright f$ can be used to test
whether $f$ holds after every $e$-step, while the formula $\tless e\tmore f$
is used to assess whether there exists an $e$-step after which $f$ holds. 
These two operators thereby follow their standard semantics
from Hennessy-Milner Logic \cite{hml}. An invariant formula $\inv{f}$ tests whether
$f$ holds in every reachable state, while a reachability expression 
$\diam{b}$ may be used to check whether there exists a path such that
the state-based formula $b$ holds at some state on this path. Note that the sub-formula $b$
of a reachability expression $\diam{b}$ is restricted to a state-based formula
$b\in\mathcal{B}$. This is used to acquire unique synthesis
solutions. The deadlock-free test $\dlf$ tests whether there exists an
outgoing step of a particular state. Combined with the invariant operator,
the formula $\inv{\dlf}$ can be used to specify that the entire synthesized
system should be deadlock-free. Validity of formulas in $\mathcal{B}$ and
$\mathcal{F}$, with respect to a Kripke-LTS $k\in\mathcal{K}$, is as shown
in Definition \ref{def:val}.

\begin{definition}
\label{def:val}
For $k=(X,L,\longrightarrow,x)\in\mathcal{K}$ and $f\in\mathcal{F}$ we define
if $k$ \emph{satisfies} $f$ (notation: $k\vDash f$) as follows: 
\end{definition}

\begin{center}
\scalebox{.88}{
\begin{tabular}{c}
$\displaystyle\frac{}{k\vDash\mathit{true}}$
\quad
$\displaystyle\frac{p\in L(x)}{(X,L,\longrightarrow,x)\vDash p}$
\quad
$\displaystyle\frac{\neg k\vDash b}{k\vDash\neg b}$
\quad
$\displaystyle\frac{k\vDash f\quad k\vDash g}{k\vDash f\wedge g}$
\quad
$\displaystyle\frac{k\vDash f}{k\vDash f\vee g}$
\quad
$\displaystyle\frac{k\vDash g}{k\vDash f\vee g}$
\bigskip \\
$\displaystyle\frac{\forall\, x\step{e}x'\quad (X,L,\longrightarrow,x')\vDash f}{(X,L,\longrightarrow,x)\vDash\bleft e\bright f}$
\quad
$\displaystyle\frac{x\step{e}x'\quad (X,L,\longrightarrow,x')\vDash f}{(X,L,\longrightarrow,x)\vDash\tless e\tmore f}$
\bigskip \\
$\displaystyle\frac{\forall\,x\longrightarrow^*x'\quad (X,L,\longrightarrow,x')\vDash f}{(X,L,\longrightarrow,x)\vDash\inv{f}}$
\quad
$\displaystyle\frac{x\longrightarrow^*x'\quad (X,L,\longrightarrow,x')\vDash b}{(X,L,\longrightarrow,x)\vDash\diam{b}}$
\quad
$\displaystyle\frac{x\step{e}x'}{(X,L,\longrightarrow,x)\vDash\dlf}$ \\
\end{tabular}}
\end{center}

We may now formulate the synthesis problem in terms of the previous
definitions in Definition \ref{def:synthesis}. Research in this paper
focuses on resolving this problem.

\begin{definition}
\label{def:synthesis}
Given $k\in\mathcal{K}$ and $f\in\mathcal{F}$, find $s\in\mathcal{K}$
in a finite method such that the following holds: 1) $s\vDash f$, 2) $s\preceq k$,
3) $s\pbis k$, 4) For all $k'\preceq k$ and $k'\vDash f$ holds $k'\preceq s$,
or determine that such an $s$ does not exist.
\end{definition}

These four properties are interpreted in the context of supervisory
control synthesis as follows. Property 1 (\emph{validity}) states that
the synthesis result satisfies the synthesized formula. Property 2 (\emph{simulation})
asserts that the synthesis result is a restriction of the original
behavior, while property 3 (\emph{controllability}) ensures that no
accessible uncontrollable behavior is disallowed during synthesis. 
Controllability is achieved if the synthesis result is related to
the original plant-model via partial bisimulation, which adds bisimulation
of all uncontrollable events to the second property. Note that the third
property implies the second property, as can be observed in Definitions
\ref{def:sim} and \ref{def:pbis}. Property 4 (\emph{maximality}) states that 
synthesis removes the least possible behavior, and thereby induces
maximal permissiveness. That is, the behavior of every alternative
synthesis option is included in the behavior of the synthesis result.

\section{Synthesis}
\label{sec:synthesis}

The purpose of this section is to illustrate the formal definition of the
synthesis construction. Synthesis as defined in this paper involves three 
major steps, after which a modified Kripke-LTS is obtained. If synthesis
is successful, the resulting structure satisfies all synthesis requirements, 
as stated in Definition \ref{def:synthesis}. The first stage of synthesis 
transforms the original transition relation 
$\longrightarrow\,\subseteq X\times\mathcal{E}\times X$, for state-space $X$, 
into a new transition relation $\longrightarrow_0\,\subseteq 
(X\times\mathcal{F})\times\mathcal{E}\times(X\times\mathcal{F})$
over the state-formula product space. This allows us to indicate
precisely which modal (sub-)formula needs to hold at each point
in the new transition relation. The second step removes transitions based 
upon an assertion of \emph{synthesizability} of formulas assigned to the 
target states of transitions. This second step is repeated until no more 
transitions are removed. The third and final synthesis step tests whether 
synthesis has been successful by evaluating whether the \emph{synthesizability} 
predicate holds for every remaining state. An overview of the synthesis process
is shown in Fig. \ref{fig:overview}.

\begin{figure}
\includegraphics[scale=.80]{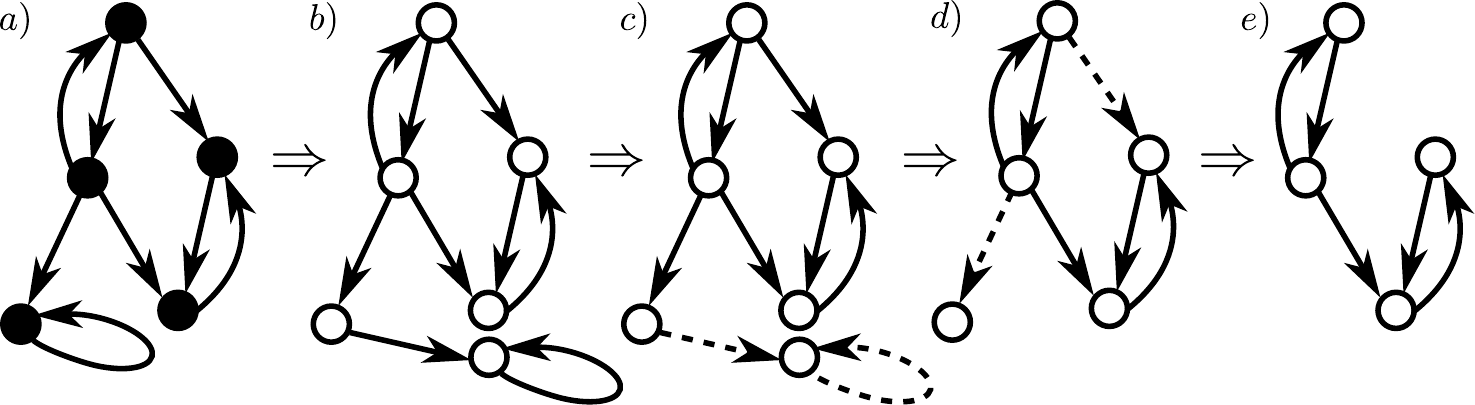}
\caption{Overview of the synthesis process. Steps in the original transition
  relation (Fig. \ref{fig:overview}a) of type $x\step{e}x'$ are combined with
  reductions of the synthesized requirement (Fig. \ref{fig:overview}b), resulting
  in transitions of type $(x,f)\step{e}_0(x',f')$, and possible inducing unfoldings. 
  Transition are then removed (Fig. \ref{fig:overview}c-\ref{fig:overview}d) based upon a local synthesizability
  test for formulas assigned to target states, until synthesizability holds in
  every reachable state (Fig. \ref{fig:overview}e).}
\label{fig:overview}
\end{figure}

A formal derivation of the starting point in the synthesis process
$\longrightarrow_0$ is shown in Definition \ref{def:zero}. This definition
relies upon the notion of sub-formulas, as formalized in Definition 
\ref{def:sub}. 

\begin{definition}
\label{def:sub}
We say that $f\in\mathcal{F}$ is a sub-formula of $g\in\mathcal{F}$
(notation $f\in\sub{g}$) if this can be derived by the following rules: 
\end{definition}
\begin{center}
\begin{tabular}{c}
$\displaystyle\frac{}{f\in\sub{f}}$
\quad
$\displaystyle\frac{f\in\sub{g}}{f\in\sub{g\wedge h}}$
\quad
$\displaystyle\frac{f\in\sub{h}}{f\in\sub{g\wedge h}}$
\quad
$\displaystyle\frac{f\in\sub{g}}{f\in\sub{\inv{g}}}$ \\
\end{tabular}
\end{center}

As shown in Definition \ref{def:sub}, sub-formulas align precisely with
the restrictions on formula expansion for conjunctive and invariant
formulas, as embedded in the the formula reductions shown in Definition
\ref{def:zero}. These restrictions on formula expansion guarantee finiteness
of formula reductions. 

\begin{definition}
\label{def:zero}
For state-space $X$ and original transition relation $\longrightarrow
\subseteq X\times\mathcal{E}\times X$, we define the starting point
of synthesis $\longrightarrow_0\subseteq (X\times\mathcal{F})\times
\mathcal{E}\times(X\times\mathcal{F})$ as follows: 
\end{definition}
\begin{center}
\scalebox{.89}{
\begin{tabular}{c}
$\displaystyle\frac{x\step{e}x'}{(x,b)\step{e}_0(x,\mathit{true})}$
\quad
$\displaystyle\frac{(x,f)\step{e}_0(x',f')\quad(x,g)\step{e}_0(x',g')\quad g'\in\sub{f'}}{(x,f\wedge g)\step{e}_0(x',f')}$
\bigskip \\
$\displaystyle\frac{(x,f)\step{e}_0(x',f')\quad(x,g)\step{e}_0(x',g')\quad g'\not\in\sub{f'}}{(x,f\wedge g)\step{e}_0(x',f'\wedge g')}$
\quad
$\displaystyle\frac{x\step{e}x'\quad x\vDash b}{(x,b\vee f)\step{e}_0(x',\mathit{true})}$
\bigskip \\
$\displaystyle\frac{(x,f)\step{e}_0(x',f')}{(x,b\vee f)\step{e}_0(x',f')}$
\quad
$\displaystyle\frac{x\step{e}x'}{(x,\bleft e\bright f)\step{e}_0(x',f)}$
\quad
$\displaystyle\frac{x\step{e}x'\quad e\not=e'}{(x,\bleft e'\bright f)\step{e}_0(x',\mathit{true})}$
\bigskip \\
$\displaystyle\frac{x\step{e}x'\quad}{(x,\tless e\tmore f)\step{e}_0(x',f)}$
\quad
$\displaystyle\frac{x\step{e}x'\quad}{(x,\tless e'\tmore f)\step{e}_0(x',\mathit{true})}$
\quad
$\displaystyle\frac{(x,f)\step{e}_0(x',f')\quad f'\in\sub{\inv{f}}}{(x,\inv{f})\step{e}_0(x',\inv{f})}$
\bigskip \\
$\displaystyle\frac{(x,f)\step{e}_0(x',f')\quad f'\not\in\sub{\inv{f}}}{(x,\inv{f})\step{e}_0(x',\inv{f}\wedge f')}$
\quad
$\displaystyle\frac{x\step{e}x'}{(x,\diam{b})\step{e}_0(x',\mathit{true})}$
\bigskip \\
$\displaystyle\frac{x\step{e}x'}{(x,\diam{b})\step{e}_0(x',\diam{b})}$
\quad
$\displaystyle\frac{x\step{e}x'}{(x,\dlf)\step{e}_0(x',\mathit{true})}$ \\
\end{tabular}}
\end{center} 

The starting point of synthesis $\longrightarrow_0$ is subjected
to transition removal via a synthesizability test for formulas assigned
to the target states of transitions. In generalized form, we define
a formula $f\in\mathcal{F}$ to be synthesizable in the state-formula
pair $(x,g)$ if this can be derived by the rules in Definition 
\ref{def:syn}. For an appropriate definition of synthesizability,
it is necessary to extend the notion of sub-formulas in such a way
that a state-based evaluation can be incorporated, in order to
handle disjunctive formulas correctly. This leads to
the sub-formula notion called $\mathit{part}$, which is shown in
Definition \ref{def:part}. 

\begin{definition}
\label{def:part}
We say that a formula $f\in\mathcal{F}$ is a part of a formula
$g\in\mathcal{F}$ in the context of a state based evaluation for
$(X,L,\longrightarrow,x)$ if this can be derived as follows:
\end{definition}
\begin{center}
\begin{tabular}{c}
$\displaystyle\frac{}{f\in\ispart{x}{f}}$
\quad
$\displaystyle\frac{f\in\ispart{x}{g}}{f\in\ispart{x}{g\wedge h}}$
\quad
$\displaystyle\frac{f\in\ispart{x}{h}}{f\in\ispart{x}{g\wedge h}}$ 
\bigskip \\
$\displaystyle\frac{x\not\vDash b\quad f\in\ispart{x}{g}}{f\in\ispart{x}{b\vee g}}$
\quad
$\displaystyle\frac{f\in\ispart{x}{g}}{f\in\ispart{x}{\inv{g}}}$ \\
\end{tabular}
\end{center}

Partial formulas as shown in Definition \ref{def:part} are used in
the definition of synthesizability as shown in Definition 
\ref{def:syn}. In particular, this is used in the definition of
synthesizability for formulas of type $\tless e\tmore f$. In addition,
partial formulas play a major role in the correctness proofs of the
synthesis method. 

\begin{definition}
\label{def:syn}
With regard to an intermediate relation $\longrightarrow_n\subseteq
(X\times\mathcal{F})\times\mathcal{E}\times(X\times\mathcal{F})$ in the
synthesis procedure, we say that a formula $f\in\mathcal{F}$ is
\emph{synthesizable} in the state-formula pair $(x,g)$ (notation:
$(x,g)\uparrow f$) if this can be derived as follows:
\end{definition}
\begin{center}
\begin{tabular}{c}
$\displaystyle\frac{x\vDash b}{(x,g)\uparrow b}$
\quad
$\displaystyle\frac{(x,g)\uparrow f_1\quad(x,g)\uparrow f_2}{(x,g)\uparrow f_1\wedge f_2}$
\quad
$\displaystyle\frac{x\vDash b}{(x,g)\uparrow b\vee f}$
\quad
$\displaystyle\frac{(x,g)\uparrow f}{(x,g)\uparrow b\vee f}$
\bigskip \\
$\displaystyle\frac{}{(x,g)\uparrow\bleft e\bright f}$
\quad
$\displaystyle\frac{(x',g')\uparrow f\quad (x,g)\step{e}_n(x',g')\quad f\in\ispart{x'}{g'}}{(x,g)\uparrow\tless e\tmore f}$
\bigskip \\
$\displaystyle\frac{(x,g)\uparrow f}{(x,g)\uparrow\inv{f}}$
\quad
$\displaystyle\frac{(x,g)\longrightarrow_n^*(x',g')\quad x'\vDash b}{(x,g)\uparrow\diam{b}}$
\quad
$\displaystyle\frac{(x,g)\step{e}_n(x',g')}{(x,g)\uparrow\dlf}$ \\
\end{tabular}
\end{center}

It is important to note here that the \emph{synthesizability} test
serves as a partial assessment. The synthesizability predicate for $f$ 
holds in the state-formula pair $(x,g)$ if it is possible to modify
outgoing transitions of $(x,g)$ in such a way that $f$ becomes 
satisfied in $(x,g)$. However, synthesizability is not straightforwardly
definable for a number of formulas. For instance, it can not be
directly assessed whether it is possible to satisfy an invariant
formula. Therefore, the synthesizability test in Definition \ref{def:syn}
is designed to operate in conjunction with the process of repeated 
transition removal, as shown in Fig. \ref{fig:overview}. This is
reflected, for instance, in the definition of synthesizability for
an invariant formula $\inv{f}$, which only relies upon $f$ being
synthesizable. However, since synthesizability needs to hold at
every reachable state for synthesis to be successful, such a 
definition of synthesizability for invariant formulas is appropriate
due to its role in the entire synthesis process. A synthesis example
for the invariant formula $\inv{p\wedge\bleft a\bright q}$ is shown
in Fig. \ref{fig:complex}. 

\begin{figure}
\includegraphics[scale=.55]{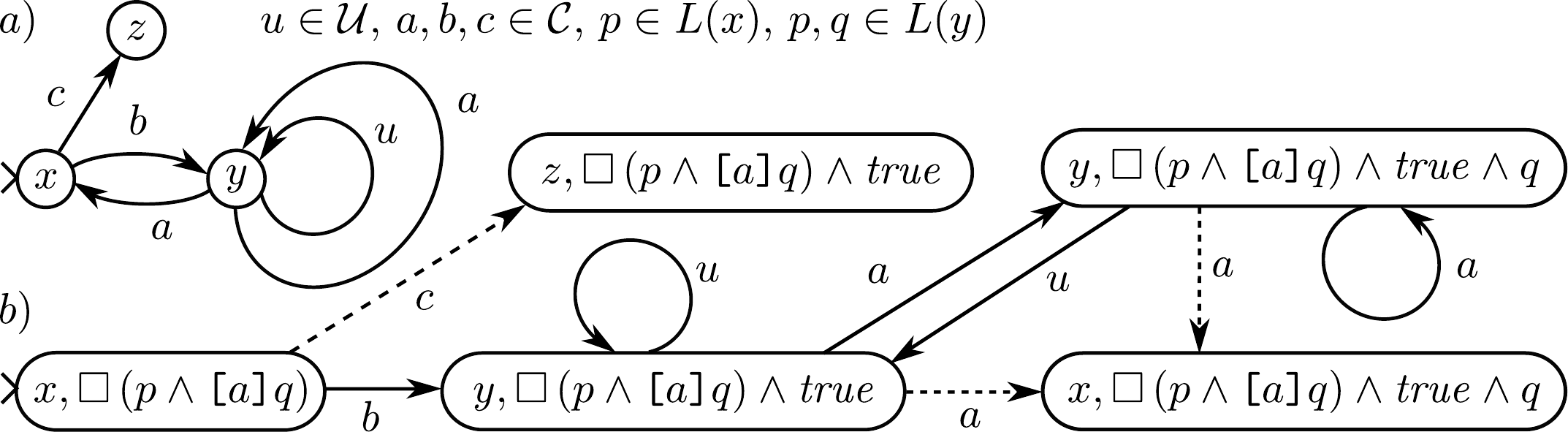}
\caption{Synthesis for the formula $\inv{p\wedge\bleft a\bright q}$ upon
  the model in Fig. \ref{fig:complex}a, resulting in the restricted
  behavioral model shown in Fig. \ref{fig:complex}b. Note the unfolding for 
  $\bleft a\bright q$, the restricted formula-expansion for invariant
  formulas, and transition disabling, indicated by dashed lines, due to 
  the state-based formula $q$ not being synthesizable in $x$, and $p$
  not being synthesizable in $z$.}
\label{fig:complex}
\end{figure}

Using the definitions stated before, we are now ready to define
the main synthesis construction. That is, how transitions are
removed from the synthesis starting point $\longrightarrow_0$,
and how are the subsequent intermediate transition relations
$\longrightarrow_1, \longrightarrow_2, \ldots$ constructed. In
addition, more clarity is required with regard to reaching a
stable point during synthesis, and verifying whether the synthesis
construction has been completed successfully. 

\begin{definition}
\label{def:construction}
For $k=(X,L,\longrightarrow,x)\in\mathcal{K}$ and $f\in\mathcal{F}$,
we define the $n$-th iteration in the synthesis construction as
follows: 
\begin{center}
\scalebox{.95}{
\begin{tabular}{c}
$\displaystyle\frac{(x,f)\step{e}_n(x',f')\quad e\in\mathcal{U}}{(x,f)\step{e}_{n+1}(x',f')}$
\quad
$\displaystyle\frac{(x,f)\step{e}_n(x',f')\quad (x,f)\uparrow f}{(x,f)\step{e}_{n+1}(x',f')}$
\bigskip \\
\end{tabular}}
\end{center}
The corresponding system model $S_{k,f}^n$ is defined as stated below,
using the labeling function $L_\mathit{proj}$, such that $L_\mathit{proj}(y,g)=L(y)$,
for all $y\in X$ and $g\in\mathcal{F}$. 
\begin{center}
$S_{k,f}^n=(X\times\mathcal{F},L_\mathit{proj},\longrightarrow_n,(x,f))$
\end{center}
\end{definition}

One last definition remains, namely \emph{completeness} of the synthesis
construction. The formula reductions induced by Definition \ref{def:zero}
are finite, which implies a terminating construction of the transition
relation $\longrightarrow_0$. Since $\longrightarrow_0$ consists of finitely
many transitions, only finitely many steps may be removed. This means that
at some point, no more transitions are removed, and a stable point will be
reached. If at this point, synthesizability holds at every reachable state,
synthesis is successful. Otherwise, it is not. It is natural that a formal
notion representing the first situation serves as a premise for a number
of correctness results. This notion is formalized as \emph{completeness} in
Definition \ref{def:complete}.

\begin{definition}
\label{def:complete}
For $k=(X,L,\longrightarrow,x)\in\mathcal{K}$, $f\in\mathcal{F}$ and $n\in\mathbb{N}$, we say that
$S_{k,f}^n$ is \emph{complete} if the following holds: 
\begin{center}
For all $(x,f)\longrightarrow_n^*(x',f')$ it holds that $(x',f')\uparrow f'$. 
\end{center}
\end{definition}

\section{Correctness}
\label{sec:correctness}
In this section, we state the theorems for the key properties 
related to synthesis correctness: termination, validity, simulation, 
controllability and maximality, as given in Definition \ref{def:synthesis}. 
All proofs are computer-verified using the Coq proof assistant \cite{sources}. 
The first result is shown in Theorem \ref{thm:term}: the synthesis 
construction is always terminating. 

\begin{theorem}
\label{thm:term}
For $k=(X,L,\longrightarrow,x)\in\mathcal{K}$, having finite $\longrightarrow$, and 
$f\in\mathcal{F}$, there exists an $n\in\mathbb{N}$ such that $S_{k,f}^n=S_{k,f}^m$ 
for all $m>n$. 
\end{theorem}
\begin{proof}
Observing the synthesis construction in Definition \ref{def:construction} it is
straightforward that from the starting point of synthesis $\longrightarrow_0$,
transitions are only removed, and not added. This means that once we are able to
show that $\longrightarrow_0\,\subseteq(X\times\mathcal{F})\times\mathcal{E}\times
(X\times\mathcal{F})$ is finite, given that $\longrightarrow$ is finite, then the
synthesis construction is terminating. In other words, only finitely many transitions
will ever be removed, if they do not satisfy the synthesizability test for the formula
assigned to the target state. The focus of this proof is therefore on the finiteness
of $\longrightarrow_0$. 

Let $\rightarrowtriangle$ denote the formula reduction relation as implicitly defined
in Definition \ref{def:zero}. That is $f\red{e}f'$ if $(x,f)\step{e}_0(x',f')$. The
reflexive-transitive closure of $\rightarrowtriangle$, denoted as $\rightarrowtriangle^*$,
is defined in the natural way. It is clear that if $(x,f)\longrightarrow_0^*(x',f')$ then
$f\rightarrowtriangle^*f'$. This means that if, for each $f$ there exists a finite set 
$D\subset\mathcal{F}$ such that for all $f\rightarrowtriangle^*f'$, $f'\in D$, then
only finitely many transitions are constructed in $\longrightarrow_0$, under the
assumption that $\longrightarrow$ is finite. We prove this property by induction
towards the structure of $f$.

If $f\equiv b$, for $b\in\mathcal{B}$, then choose $D=\{b,\,\mathit{true}\}$, which
is clearly finite. Two other cases can be handled in a similar way. For $f\equiv\diam{b}$,
choose $D=\{\diam{b},\,\mathit{true}\}$ and for $f\equiv\dlf$, choose $D=\{\dlf,\,\mathit{true}\}$.

For the case $f\equiv f_1\wedge f_2$, then by induction we
obtain two finite sets $D_1\subset\mathcal{F}$ and $D_2\subset\mathcal{F}$, containing
the formula-reducts of $f_1$ and $f_2$ respectively. If we choose $D=D_1\cup
\{f_1'\wedge f_2'\mid f_1'\in D_1,\,f_2'\in D_2\}$, then $D$ is clearly also finite.
Assume that $g\in D$ and $g\rightarrowtriangle^*g'$, then by induction towards the 
length of $g\rightarrowtriangle^*$ it is clear that $g'\in D$. Since $f_1\wedge f_2\in D$,
this completes the proof for finiteness of reductions under conjunction. For the
next case for $f\equiv b\vee f'$, for $b\in\mathcal{B}$ we obtain a $D'\subset\mathcal{F}$
representing the finiteness of the set of $f'$-reducts. Then simply choose
$D=D'\cup\{b\vee f',\,\mathit{true}\}$, which is clearly also finite. The cases
for $f\equiv\bleft e\bright f'$ and $f\equiv\tless e\tmore f'$ can be handled in
a similar way. By induction we obtain a finite set $D'\subset\mathcal{F}$ corresponding
to the formula-reducts of $f'$. For these two respective cases it is sufficient to
choose $D=D'\cup\{\bleft e\bright f',\,\mathit{true}\}$ and $D=D'\cup
\{\tless e\tmore f',\,\mathit{true}\}$. 

The case for $f\equiv\inv{f'}$, for some $f'\in\mathcal{F}$, is somewhat more involved.
Let $D'\subset\mathcal{F}$ be obtained via induction, thus containing all $f''\in\mathcal{F}$
such that $f'\rightarrowtriangle^*f''$. Assume that $D'$ is restricted such that it
strictly contains no other elements $f''$ then those which satisfy the $f'\rightarrowtriangle^*f''$
condition. We then define the function $d:\mathcal{F}\times2^\mathcal{F}\times\mathbb{N}\mapsto2^\mathcal{F}$ 
in the following way:

\begin{center}
\begin{math}
\begin{array}{lcl}
d(f,D,0)     & = & \{\inv{f}\} \\
d(f,D,n + 1) & = & d(f,D,n)\cup\{f'\wedge g'\mid f'\in d(f,D,n),\, g'\in D\} \\
\end{array}
\end{math}
\end{center}

As our witness, we then choose $D=d(f',D',\mid D'\mid)$, where $\mid D'\mid$ refers to the
finite number of elements in $D'$. Clearly it holds that $\inv{f'}\in D$, by the definition
of $d$. For each $f'\rightarrowtriangle^*f''$, there exists an $n\in\mathbb{N}$, such that
$f''\in d(f',D',n)$. However, the application of $\mathit{sub}$ in the formula reductions
for conjunction and invariant formulas in Definition \ref{def:zero} ensure that if 
$g\in d(f',D',n)$ then $n\leq\mid D'\mid$, as can be derived using induction towards
$\mid D'\mid$. 
\end{proof}

The second result is shown in Theorem \ref{thm:val}: If synthesis is 
\emph{complete} then the synthesis result satisfies the synthesized 
formula. Since synthesis is terminating, as shown in Theorem 
\ref{thm:term}, this results in a stable point in the synthesis
process. It may then be quickly assessed whether synthesis is complete,
by checking whether synthesizability is satisfied in every remaining
reachable state, upon which the result in Theorem \ref{thm:val} holds.

\begin{lemma}
\label{lem:part_syn}
If $(x,h)\uparrow g$ and $f\in\ispart{x}{g}$ then $(x,h)\uparrow f$.
\end{lemma}
\begin{proof}
By induction towards the derivation of $(x,h)\uparrow g$, using Definition
\ref{def:syn}.
\end{proof}

\begin{theorem}
\label{thm:val}
If $S_{k,f}^n$ is complete then $S_{k,f}^n\vDash f$.
\end{theorem}
\begin{proof}
Assume $k=(X,L,\longrightarrow,x)$ and $g\in\mathcal{F}$. We show a more 
generalized result: if $f\in\ispart{x}{g}$ and $S_{k,g}^n$ is complete then 
$S_{k,g}^n\vDash f$. This immediately leads to the required result, since 
$f\in\ispart{x}{f}$. Note that we have $(x,g)\uparrow g$ by Definition
\ref{def:complete} and due to $(x,g)\longrightarrow^*_n(x,g)$. Also, we
have $(x,g)\uparrow f$, by Lemma \ref{lem:part_syn}.

Apply induction towards the structure of $f$. Suppose that $f\equiv b$, for
some $b\in\mathcal{B}$. Then from $(x,g)\uparrow b$ we have $x\vDash b$,
which directly leads to $S_{k,g}^n\vDash b$, due to the fact that validity
of a state-based formula $b\in\mathcal{B}$ only depends upon the labels
assigned to $x$. 

If $f\equiv f_1\wedge f_2$, and $f_1\wedge f_2\in\ispart{x}{g}$, then
$f_1\in\ispart{x}{g}$ and $f_2\in\ispart{x}{g}$, as is clear from Definition
\ref{def:part}. By induction, we then have $S_{k,g}^n\vDash f_1$ and
$S_{k,g}^n\vDash f_2$. For the next case, suppose that $f\equiv b\vee f'$.
If $x\vDash b$, then $S_{k,g}^n\vDash b$. However, if $x\not\vDash b$, then
$(x,g)\uparrow b$ also does not hold, so $(x,g)\uparrow f'$ must be true.
In addition, we have $f'\in\ispart{x}{g}$. This is precisely the reason
why it is necessary to incorporate a state-based evaluation in Definition
\ref{def:part}. Application of the induction hypothesis now gives $S_{k,g}^n\vDash f'$.

Suppose that $f\equiv\bleft e\bright f'$, and assume that 
$(x,g)\step{e}_n(x',g')$. Using Definition \ref{def:part}, we may then
conclude that $f'\in\ispart{x'}{g'}$. Let $k'=(X,L,\longrightarrow,x')$.
We apply induction in order to obtain $S_{k',g'}^n\vDash f'$. Due to 
the assumption of $(x,g)\step{e}_n(x',g')$, the induction premise for
completeness is satisfied for $S_{k',g'}^n$ as well. If $f\equiv\tless e\tmore f'$,
then by Lemma \ref{lem:part_syn} we have $(x,g)\uparrow\tless e\tmore f'$.
By Definition \ref{def:syn}, there now exists a step $(x,g)\step{e}_n(x',g')$
such that $f'\in\ispart{x'}{g'}$. The latter condition shows why it is
important to have the condition $f'\in\ispart{x'}{g'}$ in Definition
\ref{def:syn}, for the formula $\tless e\tmore f$. We apply the induction
hypothesis to derive $S_{k',g'}^n\vDash f'$, for $k'=(X,L,\longrightarrow,x')$.
Again, the induction premise for completeness in $S_{k',g'}^n$ is satisfied
due to the existence of the step $(x,g)\step{e}(x',g')$, and completeness
of $S_{k,g}^n$. 

The next case considers $f\equiv\inv{f'}$, for some $f'\in\mathcal{F}$.
Assume the existence of a step-sequence $(x,g)\longrightarrow_n^*(x',g')$.
By Definitions \ref{def:part} and \ref{def:zero}, it is clear that
$\inv{f'}\in\ispart{x'}{g'}$, and therefore $f'\in\ispart{x'}{g'}$. This
allows us to apply the induction hypothesis for $f'$, in order to obtain
$S_{k',g'}^n\vDash f'$ for each $(x,g)\longrightarrow_n^*(x',g')$ and
$k'=(X,L,\longrightarrow,x')$. Hence, we obtain $S_{k,g}^n\vDash\inv{f'}$.

Suppose that $f\equiv\diam{b}$, for some $b\in\mathcal{B}$. By Lemma
\ref{lem:part_syn}, there exists a path $(x,g)\longrightarrow_n^*(x',g')$
such that $x'\vDash b$, leading directly to $S_{k,g}^n\vDash\diam{b}$.
For $f\equiv\dlf$, the derivation $(x,g)\uparrow\dlf$ from Lemma
\ref{lem:part_syn} also leads directly to $S_{k,g}^n\vDash\dlf$.
\end{proof}

We show that our synthesis method adheres to controllability by
verifying that the synthesis result is related to the original
plant model via partial bisimulation in Theorem \ref{thm:pbis}.
Note that this implies simulation.

\begin{lemma}
\label{lem:syn_impl_rel}
If $(x,f)\uparrow f$ and $x\step{e}x'$ and $e\in\mathcal{U}$, then there exists
an $f'\in\mathcal{F}$ such that for all $n\in\mathbb{N}$, we have $(x,f)\step{e}_n(x',f')$.
\end{lemma}
\begin{proof}
Using induction towards the structure of $f$, we may derive the existence of an $f'\in\mathcal{F}$,
such that $(x,f)\step{e}_0(x',f')$. Given that $e\in\mathcal{U}$, it is then straightforwardly
derivable that $(x,f)\step{e}_n(x',f')$, by induction on $n$.
\end{proof}

\begin{theorem}
\label{thm:pbis}
If $S_{k,f}^n$ is complete then $S_{k,f}^n\pbis k$.
\end{theorem}
\begin{proof}
Let $k=(X,L,\longrightarrow,x)$. According to Definition \ref{def:pbis}, 
we need to provide a witness relation $R$, such that $S_{k,f}^n\pbis_R k$. 
Choose $R=\{((y,g),y)\mid S_{k',g}^n$ is complete, for $k'=(X,L,\longrightarrow,y)\}$. 
Suppose that $((y,g),y)\in R$. If there exists a step $(y,g)\step{e}_n(y',g')$,
then by Definition \ref{def:zero}, there also exists a step $y\step{e}y'$,
upon which we may conclude that $((y',g'),y')\in R$, since completeness
of $S_{k',g}^n$ extends to completeness of $S_{k'',g'}^n$, for $k''=(X,L,\longrightarrow,y')$.
If $y\step{e}y'$, for $e\in\mathcal{U}$, then by Lemma \ref{lem:syn_impl_rel},
there exists a $g'\in\mathcal{F}$, such that $(y,g)\step{e}_n(y',g')$, which
again leads to the conclusion that $((y',g'),y')\in R$. Note that the premise $(y,g)\uparrow g$
in Lemma \ref{lem:syn_impl_rel}, is derived from the completeness of $S_{k',g}^n$. 
\end{proof}

As a final result, we show that the synthesis result is maximal
within the simulation preorder, with respect to all simulants of
the original system which satisfy the synthesized formula. This
result, which implies maximal permissiveness in the context of
supervisory control, is shown in Theorem \ref{thm:max}.

\begin{lemma}
\label{lem:step_val}
For $f\in\mathcal{F}$, $k'=(X',L',\longrightarrow',x')$ and 
$k=(X,L,\longrightarrow,x)$ such that $k'\preceq k$ and $k'\vDash f$,
and if $x'\primestep{e}y'$ and $x\step{e}y$, there exists an $f'\in\mathcal{F}$
such that $(X',L',\longrightarrow',y')\vDash f'$ and $(x,f)\step{e}_0(y,f')$.
\end{lemma}
\begin{proof}
By induction towards the structure of $f$. Note that simulation as
given in Definition \ref{def:sim} includes strict equivalence of labels
in related states. This implies that validity of a formula $b\in\mathcal{B}$
is preserved under simulation. This fact must be used in order to derive
existence of a step for a disjunctive formula in the induction argument for $f$.
\end{proof}

\begin{lemma}
\label{lem:rel_ex}
If $(x,f)\step{e}_0(x',f')$ and if $(x',f')\uparrow f'$, with relation
to $\longrightarrow_n$, then $(x,f)\step{e}_n(x',f')$.
\end{lemma}
\begin{proof}
Note that the premise $(x',f')\uparrow f'$, with relation to $\longrightarrow_n$,
should be interpreted as if $\longrightarrow_n$ were applied in Definition \ref{def:syn}.
Apply induction towards $n$. If $n\equiv 0$ then it is clear that $(x,f)\step{e}_0(x',f')$.
Suppose that $(x,f)\step{e}_n(x',f')$ for some $n\in\mathbb{N}$, and $(x',f')\uparrow f'$,
with relation to $\longrightarrow_{n+1}$. Then also $(x',f')\uparrow f'$ with relation
to $\longrightarrow_n$. Then, by Definition \ref{def:construction}, it is clear that
all conditions for the derivation of $(x,f)\step{e}_{n+1}(x',f')$ are satisfied.
\end{proof}

\begin{lemma}
\label{lem:part_val_syn}
If $k=(X,L,\longrightarrow,x)$ and $k'=(X',L',\longrightarrow',x')$ such
that $k'\preceq k$, and if $f\in\ispart{x}{g}$ such that $k'\vDash g$, then
$(x,g)\uparrow f$, with regard to $\longrightarrow_n$. 
\end{lemma}
\begin{proof}
The proof of this lemma is somewhat complicated, and involves induction
towards $n$, and within this induction argument, induction towards the
structure of $f$. For both the inductive cases $n\equiv 0$ and $n+1$, a
number of cases for $f$ may be resolved directly. This applies to the
cases for $f\equiv b$, for $b\in\mathcal{B}$, and $f\equiv\bleft e\bright f'$.
The case for $f\equiv f_1\wedge f_2$ can be resolved via solely the
induction hypotheses for $f_1$ and $f_2$. If $f\equiv b\vee f'$, then
we distinguish between the situations where $x\vDash b$ and $x\not\vDash b$.
In the latter case, we have $f'\in\ispart{x}{g}$, which allows us to
apply the induction hypothesis for $f'$, in order to derive $(x,g)\uparrow f'$. 
If $f\equiv\inv{f'}$, for some $f'\in\mathcal{F}$, then by Definition \ref{def:syn},
we only need to derive $(x,g)\uparrow f'$, which is straightforward using
the induction hypothesis for $f'$. 

The remaining cases for $f\equiv\tless e\tmore f'$, $f\equiv\diam{b}$ and
$f\equiv\dlf$ are somewhat more involved, and rely upon the induction
hypothesis for $n$. Since $k'\vDash g$, and $\tless e\tmore f'\in\ispart{x}{g}$,
it is clear that $k'\vDash\tless e\tmore f'$. This means that there exists
a step $x'\primestep{e}y'$ such that $(X',L',\longrightarrow',y')\vDash f'$.
Since $k'\preceq k$, and because $x'\primestep{e}y'$, there exists a step
$x\step{e}y$, such that $(X',L',\longrightarrow',y')\preceq (X,L,\longrightarrow,y)$.
Using a separate argument, by induction towards the derivation of $\tless e\tmore f'
\in\ispart{x}{g}$, there exists a $g'\in\mathcal{F}$ such that $(X',L',\longrightarrow',y')
\vDash g'$, $f'\in\ispart{y}{g}$ and $(x,g)\step{e}_0(y,g')$. We may then apply
Lemma \ref{lem:rel_ex}, in order to construct a step $(x,g)\step{e}_n(y,g')$ such
that $(y,g')\uparrow f'$ by induction. The premise $(y,g')\uparrow g'$, with
regard to $\longrightarrow_n$, is obtained by the induction hypothesis for $n$.

The next case to consider is $f\equiv\diam{b}$, for some $b\in\mathcal{B}$. 
The construction used here is somewhat similar to the previously applied
construction for the case $f\equiv\tless e\tmore f'$. Since $k'\vDash g$ and
$\diam{b}\in\ispart{x}{g}$, there exists a path $x'[\longrightarrow']^*z'$
such that $z'\vDash b$. We apply induction towards the length of
$x'[\longrightarrow']^*z'$ in order to derive a path $(x,g)\longrightarrow_n^*(z,g'')$,
such that $y\vDash b$. This is sufficient to derive $(x,g)\uparrow\diam{b}$, as
shown in Definition \ref{def:syn}. Just as in the case for $f\equiv\tless e\tmore f'$,
we use a separate argument to derive a step $(x,g)\step{e}_0(y,g')$ such that
$(X',L',\longrightarrow',y')\vDash g'$ and $\diam{b}\in\ispart{y}{g'}$. We
then apply Lemma \ref{lem:rel_ex} and the induction hypothesis for $n$ to
construct a step $(x,g)\step{e}_n(y,g')$. Since we are applying induction 
towards the length of $x'[\longrightarrow']^*z'$, we may repeat this argument
in order to construct a path $(x,g)\longrightarrow_n(z,g'')$ such that $z\vDash b$.
This allows us to conclude that $(x,g)\uparrow\diam{b}$.

The remaining case for $f\equiv\dlf$ also relies upon application of Lemma
\ref{lem:rel_ex} and the induction hypothesis for $n$. Since $\dlf\in\ispart{x}{g}$,
it is clear that $k'\vDash\dlf$ and there exists a step $x'\step{e}y'$. By
simulation, there also exists a step $x\step{e}y$ and by Lemma \ref{lem:step_val},
there exists a $g'\in\mathcal{F}$ such that $(x,g)\step{e}_0(y,g')$. As said, application of
Lemma \ref{lem:rel_ex} and the induction hypothesis for $n$ results in 
$(x,g)\step{e}_n(y,g')$, which is sufficient to derive $(x,g)\uparrow\dlf$, by
Definition \ref{def:syn}.
\end{proof}

\begin{theorem}
\label{thm:max}
If $k'\preceq k$ and $k'\vDash f$ then $k'\preceq S_{k,f}^n$.
\end{theorem}
\begin{proof}
Suppose that $k=(X,L,\longrightarrow,x)$ and $k'=(X',L',\longrightarrow',x')$ such
that $R\subseteq X'\times X$ and $k'\preceq_R k$. As can be observed from Definition
\ref{def:sim}, we need to provide a witness $R'\subseteq X'\times (X\times\mathcal{F})$
such that $k'\preceq_{R'} S_{k,f}^n$. Choose $R'=\{(y',(y,g))\mid (y',y)\in R$ and 
$(X',L',\longrightarrow',y')\vDash g\}$. 

Suppose that $(y',(y,g))\in R'$ such that $(X',L',\longrightarrow',y')\vDash g$ and 
$(y',y)\in R$. If there exists a step $y'\primestep{e}z'$, then by Definition
\ref{def:sim}, and due to the fact that $(y',y)\in R$, there exists a step $y\step{e}z$
such that $(z',z)\in R$. We then apply Lemma \ref{lem:step_val} to obtain a $g'\in\mathcal{F}$
such that $(X',L',\longrightarrow',z')\vDash g'$ and $(y,g)\step{e}_0(z,g')$. By Lemma
\ref{lem:part_val_syn}, we derive $(z,g')\uparrow g'$. This allows us to apply
Lemma \ref{lem:rel_ex} to construct a step $(y,g)\step{e}_n(z,g')$. Upon which we
may conclude that $(z',(z,g'))\in R'$. Note that a premise for completeness is not
required in this theorem for maximal permissiveness, since this property is retained
during synthesis.
\end{proof}

\section{Conclusions}
\label{sec:conclusions}
This paper presents a novel approach to controlled system synthesis
for modal logic on non-deterministic plant models. The behavior of
a Kripke-structure with labeled transitions is adapted such that it
satisfies the synthesized requirement. The relationship between the
synthesis result and the original plant specification adheres to
important notions from Ramadge-Wonham supervisory control: controllability
and maximal permissiveness. The requirement specification logic also allows
expressibility of deadlock-freeness and marker-state reachability.
The synthesis approach, via a reduction on modal expressions combined
with an iteratively applied synthesizability test for formulas
assigned to target states of transitions results in an effective synthesis
procedure. Our next research efforts will focus on determining the
effectiveness of this procedure as well as its applicability in case
studies.

\bibliographystyle{plain}
\bibliography{complete}
\end{document}